\documentclass[12pt]{amsart}

\usepackage{amsmath,amsthm,txfonts}

\renewcommand{\Im}{\operatorname{Im}}
\newcommand{\Res}{\operatorname{Res}}
\newcommand{\ga}{\gamma}
\newcommand{\sn}{\operatorname{sn}}

\newcommand{\RR}{\mathbb{R}}

\newcommand{\ka}{\kappa}

\theoremstyle{plain}
\newtheorem{theorem}{Theorem}

\newtheorem{proposition}[theorem]{Proposition}

\theoremstyle{definition}

\newtheorem{remark}[theorem]{Remark}

\usepackage{color}

\begin{document}

\title{On symmetric primitive potentials}

\author[P.~Nabelek]{Patrik Nabelek}
\address{Department of Mathematics, Oregon State University, Corvallis, OR}
\email{patrik@alyrica.net}
\author[D.~Zakharov]{Dmitry Zakharov}
\address{Department of Mathematics, Central Michigan University, Mount Pleasant, MI}
\email{dvzakharov@gmail.com}
\author[V.~Zakharov]{Vladimir Zakharov}
\address{Department of Mathematics, University of Arizona, Tucson, AZ 85721 \\ L.~D.~Landau Institute for Theoretical Physics, Chernogolovka, Moscow Region, Russian Federation }
\email{zakharov@math.arizona.edu}

\begin{abstract} The concept of a primitive potential for the Schr\"odinger operator on the line was introduced in \cite{DZZ16,ZDZ16,ZZD16}. Such a potential is determined by a pair of positive functions on a finite interval, called the dressing functions, which are not uniquely determined by the potential. The potential is constructed by solving a contour problem on the complex plane. In this paper, we consider a reduction where the dressing functions are equal. We show that in this case, the resulting potential is symmetric, and describe how to analytically compute the potential as a power series. In addition, we establish that if the dressing functions are both equal to one, then the resulting primitive potential is the elliptic one-gap potential.

\begin{description}
\item[Keywords] integrable systems, Schr\"odinger equation, primitive potentials
\end{description}

\end{abstract}

\maketitle




\section{Introduction}

One of the fundamental insights underlying the modern theory of integrable systems is the discovery of an intimate relationship between certain linear differential or difference operators, on one hand, and corresponding nonlinear equations on the other. The first of these relationships to be discovered, and arguably the most important one, is the link between the one-dimensional Schr\"odinger equation on the real axis
\begin{equation}
-\psi''+u(x)\psi=E\psi,\quad -\infty<x<\infty, \label{eq:Sch}
\end{equation}
and the Korteweg--de Vries equation
\begin{equation}
u_t(x,t)=6u(x,t)u_x(x,t)-u_{xxx}(x,t).
\label{eq:KdV}
\end{equation}
The study of solutions of the KdV equation has proceeded hand-in-hand with an analysis of the spectral properties of the Schr\"odinger operator that is applied to $\psi$ on the left hand side of the Schr\"odinger equation (\ref{eq:Sch}).

There are three broad methods for constructing solutions of the KdV equation, based on restricting the potentials of the Schr\"odinger operator. The inverse scattering method (ISM) allows us to construct potentials, and hence solutions of the KdV equation, that are rapidly vanishing as $x\to \pm \infty$. Such potentials have a finite discrete spectrum for $E<0$ and a doubly degenerate continuous spectrum for $E>0$, and a subset of them, corresponding to multisoliton solutions of the KdV equation, are reflectionless for positive energies. The finite-gap method, on the other hand, constructs periodic and quasi-periodic potentials of the Schr\"odinger operator \eqref{eq:Sch} whose spectrum consists of finitely many allowed bands, one infinite, separated by forbidden gaps. These potentials are reflectionless in the allowed bands.

Both of these methods construct globally defined solutions of the KdV equation. The third method, called the dressing method \cite{ZakharovManakov}, constructs solutions locally near a given point on the $(x,t)$-plane. An advantage of the method is that the constructed solutions can be quite general. However, the problem of extending such solutions to the entire $(x,t)$-plane is a difficult one.

Our work is motivated by a pair of related questions. First, one can ask what is the exact relationship between the ISM and the finite-gap method, and whether they can both be generalized by the dressing method. It has long been known that multisoliton solutions of the KdV equation are limits of finite-gap solutions corresponding to rational degenerations of the spectral curve. However, the converse relationship, which would consist in obtaining finite-gap solutions as limits of multisoliton solutions, has not been worked out. Additionally, one can ask which potentials of the Schr\"odinger operator, other than the finite-gap ones, have a band-like structure.

In the papers \cite{DZZ16,ZDZ16,ZZD16}, the second and third authors presented a method for constructing potentials of the Schr\"odinger operator \eqref{eq:Sch}, called {\it primitive potentials}, that provides partial answers to these questions. Primitive potentials are constructed by directly implementing the dressing method, and can be thought of as the closure of the set of multisoliton potentials. This procedure involves a reformulation of the ISM that is inherently symmetric with respect to the involution $x\to -x$, and the resulting primitive potentials are non-uniquely determined by a pair of positive, H\"older-continuous functions, called the {\it dressing functions}, defined on a finite interval. 

In this paper we continue the study of primitive potentials. We consider primitive potentials defined by a pair of dressing functions that are equal. Such potentials are symmetric with respect to the reflection $x\to -x$. We show that the contour problem defining symmetric primitive potentials can be solved analytically, and we give an algorithm for computing the Taylor coefficients of a primitive potential. In the case when the dressing functions are both identically equal to 1, we show that the corresponding primitive potential is the elliptic one-gap potential.

\section{Primitive potentials}



In this section, we recall the definition of primitive potentials, which were first introduced in the papers \cite{DZZ16,ZDZ16,ZZD16} as generalizations of finite-gap potentials. Primitive potentials are constructed by taking the closure of the set of $N$-soliton potentials as $N\to \infty$, so we begin by summarizing the inverse scattering method (ISM) as a contour problem (see \cite{ZMNP80}, \cite{GT09}). The finite-gap method is symmetric with respect to the transformation $x\to -x$, while the ISM is not, so we give an alternative formulation of the ISM (in the reflectionless case) that takes this symmetry into account.


\subsection{The inverse scattering method}
Consider the self-adjoint Schr\"odinger operator
\begin{equation}
L(t)=-\frac{d^2}{dx^2}+u(x,t)
\end{equation}
on the Sobolev space $H^2(\RR)\subset L^2(\RR)$. We suppose that the potential $u(x,t)$ rapidly decays at infinity when $t=0$:
\begin{equation}
\int_{-\infty}^{\infty}(1+|x|)(|u(x,0)|+|u_x(x,0)|+|u_{xx}(x,0)|+|u_{xxx}(x,0)|)\,dx<\infty \label{eq:udecay}
\end{equation}
and satisfies the KdV equation \eqref{eq:KdV}. Under this assumption, the spectrum of $L(t)$ consists of an absolutely continuous part $[0,\infty)$ and a finite number of eigenvalues $-\ka_1^2,\ldots,-\ka_N^2$ that do not depend on $t$. There exist two Jost solutions $\psi_{\pm}(k,x,t)$ such that
\begin{equation}
L(t)\psi_{\pm}(k,x,t)=k^2\psi_{\pm}(k,x,t),\quad \Im(k)>0, \label{eq:Jost}
\end{equation}
with asymptotic behavior
\begin{equation}
\lim_{x\to \pm\infty} e^{\mp ikx}\psi_{\pm}(k,x,t)=1.
\end{equation}
The Jost solutions $\psi_{\pm}$ are analytic for $\Im k>0$ and continuous for $\Im k\geq 0$, and have the following asymptotic behavior as $k\to \infty$ with $\Im k>0$:
\begin{equation}
\psi_{\pm}(k,x,t)=e^{\pm ikx}\left(1+Q_{\pm}(x,t)\frac{1}{2ik}+O\left(\frac{1}{k^2}\right)\right),
\end{equation}
where
\begin{equation}
Q_+(x,t)=-\int_x^{\infty} u(y,t)\, dy,\quad Q_-(x,t)=-\int_{-\infty}^x u(y,t)\, dy.
\end{equation}



The Jost solutions satisfy the scattering relations
\begin{equation}
T(k)\psi_{\mp}(k,x,t)=\overline{\psi_{\pm}(k,x,t)}+R_{\pm}(k,t)\psi_{\pm}(k,x,t),\quad k\in \RR,
\end{equation}
where $T(k)$ and $R_{\pm}(k,t)$ are the transmission and reflection coefficients, respectively. These coefficients satisfy the following properties:

\begin{proposition} The transmission coefficient $T(k)$ is meromorphic for $\Im k>0$ and is continuous for $\Im k\geq 0$. It has simple poles at $i\ka_1,\ldots,i\ka_N$ with residues
\begin{equation}
\Res_{i\ka_j} T(k)=i\mu_j(t)\ga_{j} (t)^2,
\end{equation}
where
\begin{equation}
\ga_{j}(t)^{-1}=||\psi_+(i\ka_j,x,t)||_2,\quad \psi_+(i\ka_j,x,t)=\mu_j(t)\psi_-(i\ka_j,x,t).
\end{equation}
Furthermore,
\begin{equation}
T(k)\overline{R_+(k,t)}+\overline{T(k)}R_-(k,t)=0,\quad |T(k)|^2+|R_{\pm}(k,t)|^2=1.
\end{equation}
If we denote $R(k,t)=R_+(k,t)$, $R(k)=R(k,0)$, and $\ga_j=\ga_j(0)$, then
\begin{equation}
T(-k)=\overline{T(k)},\quad R(-k)=\overline{R(k)},\quad k\in \RR, \label{eq:symmetryofreflection},
\end{equation}
\begin{equation}
|R(k)|<1\mbox{ for }k\neq 0,\quad R(0)=-1\mbox{ if }|R(0)|=1,
\end{equation}
and the function $R(k)$ is in $C^2(\RR)$ and decays as $O(1/|k|^3)$ as $|k|\to \infty$. The time evolution of the quantities $R(k,t)$ and $\ga_j(t)$ is given by
\begin{equation}
R(k,t)=R(k)e^{8ik^3t},\quad \ga_j(t)=\ga_j e^{4\ka_j^3 t}.
\end{equation}

\end{proposition}

The collection $\left(R(k,t), k\geq 0; \ka_1,\ldots,\ka_N,\ga_1(t),\ldots,\ga_N(t)\right)$ is called the scattering data of the Schr\"odinger operator $L(t)$. We encode the scattering data as a contour problem in the following way. Consider the function
\begin{equation}
\chi(k,x,t)=\left\{\begin{array}{cc} T(k) \psi_-(k,x,t)e^{ikx}, & \Im k>0, \\ \psi_+(-k,x,t)e^{ikx}, & \Im k<0.\end{array}\right.\label{eq:chi}
\end{equation}

\begin{proposition} Let $(R(k); \ka_1,\ldots,\ka_N,\ga_1,\ldots,\ga_N)$ be the scattering data of the Schr\"odinger operator $L(0)$. Then the function $\chi(k,x,t)$ defined by \eqref{eq:chi} is the unique function satisfying the following properties: \label{prop:ISM}

\begin{enumerate}

\item $\chi$ is meromorphic on the complex $k$-plane away from the real axis and has non-tangential limits
\begin{equation}
\chi_{\pm}(k,x,t)=\lim_{\varepsilon \to 0}\chi(k\pm i\varepsilon,x,t),\quad k\in \RR
\end{equation}
on the real axis.

\item $\chi$ has a jump on the real axis satisfying
\begin{equation}
\chi_+(k,x,t)-\chi_-(k,x,t)=R(k)e^{2ikx+8ik^3t} \chi_-(-k,x). \label{eq:chijumponrealaxis}
\end{equation}

\item $\chi$ has simple poles at the points $i\ka_1,\ldots,i\ka_n$ and no other singularities. The residues at the poles satisfy the condition
\begin{equation}
\Res_{i\ka_j}\chi(k,x,t)=ic_j e^{-2\ka_jx+8\ka_j^3t}\chi(-i\ka_j,x,t),\quad c_j=\ga_j^2.
\label{eq:residues}
\end{equation}

\item $\chi$ has the asymptotic behavior
\begin{equation}
\chi(k,x,t)=1+\frac{i}{2k}Q(x,t)+O\left(\frac{1}{k^2}\right),\quad |k|\to \infty,\quad \Im k\neq 0.
\label{eq:asymptotic1}
\end{equation}
\end{enumerate}
The function $\chi$ is a solution of the equation
\begin{equation}
\chi''-2ik\chi'-u(x)\chi'=0, 
\end{equation}
and the function $u(x,t)$ given by the formula
\begin{equation}
u(x,t)=\dfrac{d}{dx}Q(x,t)\label{eq:uintermsofQ}
\end{equation}
is a solution of the KdV equation \eqref{eq:KdV} satisfying condition \eqref{eq:udecay}.
\end{proposition}

\begin{remark} We note that the contour problem for $\chi$ is not symmetric with respect to the transformation $k\to -k$. The reflection coefficient $R(k)$ satisfies the symmetry condition \eqref{eq:symmetryofreflection}, however, $\chi$ is required to have poles in the upper $k$-plane and be analytic in the lower $k$-plane. This asymmetry comes from the definition \eqref{eq:Jost} of the Jost functions and is therefore ultimately of physical origin: in the ISM, we consider a quantum-mechanical particle approaching the localized potential from the {\it right}, in other words the method is not symmetric with respect to the transformation $x\to -x$. We will see in the next section that this asymmetry prevents us from directly relating the ISM to the finite-gap method.

It is common (see \cite{GT09}) to instead consider the two-component vector $[\chi(k) \; \chi(-k)]$. The jump condition on the real axis \eqref{eq:chijumponrealaxis} is then replaced by a local Riemann--Hilbert problem. This Riemann--Hilbert problem includes poles on the upper and lower $k$-planes, but the transformation $k\to -k$ merely exchanges the components, which does not fix the asymmetry.

\label{rem:asymmetry}
\end{remark}

\begin{remark} It is possible to relax the constraint $|R(k)|<1$ for $k\neq 0$ and allow $|R(k)|$ to be equal to 1 inside two symmetric finite intervals $v<|k|<u$. In this case, the Riemann--Hilbert problem~\eqref{eq:chijumponrealaxis} is still uniquely solvable and generates a potential of the Schr\"odinger operator and a solution of the KdV equation. However, in this case condition~\eqref{eq:udecay} is not satisfied, and the potential is not rapidly decaying, at least when $x\to -\infty$. This extremely interesting case is completely unexplored. 

\end{remark}

\subsection{$N$-soliton solutions}

We now restrict our attention to the reflectionless case, in other words we assume that $R(k)=0$. In this case, the function $\chi$ has no jump on the real axis and is meromorphic on the entire $k$-plane with simple poles at the points $i\ka_1,\ldots,i\ka_N$. Hence Prop.~\ref{prop:ISM} reduces to the following. 

\begin{proposition} Let $(0; \ka_1,\ldots,\ka_N,\ga_1,\ldots,\ga_N)$ be the scattering data of the Schr\"odinger operator $L(0)$ with zero reflection coefficient. Then the function $\chi(k,x,t)$ defined by \eqref{eq:chi} is the unique function satisfying the following properties:

\begin{enumerate}
\item $\chi$ is meromorphic on the complex $k$-plane with simple poles at the points $i\ka_1,\ldots,i\ka_N$ and no other singularities, and its residues satisfy condition~\eqref{eq:residues}.

\item $\chi$ has the asymptotic behavior \eqref{eq:asymptotic1} as $|k|\to \infty$.
\end{enumerate}
%

\end{proposition}

The corresponding solution $u(x,t)$ of the KdV equation \eqref{eq:KdV}, given by formula~\eqref{eq:uintermsofQ}, is known as the $N$-soliton solution. Finding this solution is a linear algebra exercise. If $\chi$ is expressed in terms of its residues
\begin{equation}
\chi=1+\sum_{n=1}^N\frac{\chi_n}{k-i\ka_n},
\end{equation}
then plugging this into equation~\eqref{eq:residues} gives a linear equation
\begin{equation}
\chi_n+c_ne^{-2\ka_nx+8\ka_n^3t}\sum_{m=1}^N\frac{\chi_m}{\ka_n+\ka_m}=c_ne^{-2\ka_nx+8\ka_n^3t}. \label{eq:chiresidues}
\end{equation}
Let $A$ be the determinant of this system:
\begin{equation}
A=\sum_{I\subset \{1,\ldots,N\}}\prod_{(i,j)\subset I,\,i<j} \frac{(\ka_i-\ka_j)^2}{(\ka_i+\ka_j)^2}\prod_{i\in I}q_i e^{-2\ka_i x+8\ka_i^3t},\quad q_i=\frac{c_i}{2\ka_i}>0. \label{eq:NsolitonA}
\end{equation}
Then the corresponding $N$-soliton solution of the KdV equation~\eqref{eq:KdV} is
\begin{equation}
u(x,t)=-2\frac{d^2}{dx^2} \log A.\label{eq:Nsoliton}
\end{equation}

\subsection{The na\"ive limit $N\to \infty$} The papers \cite{DZZ16}, \cite{ZDZ16}, \cite{ZZD16} were motivated by the following question. There exists a family of solutions of the KdV equation, called the finite-gap solutions, that are parametrized by the data of a hyperelliptic algebraic curve with real branch points and a line bundle on it. The solutions are given by the Matveev--Its formula
\begin{equation}
u(x,t)=-2\frac{d^2}{d x^2}\ln \Theta(Ux+Vt+Z|B), \label{eq:MatveevIts}
\end{equation}
where $\Theta(\cdot|B)$ is the Riemann theta function of the hyperelliptic curve, and $U$, $V$, and $Z$ are certain vectors. The solution $u(x,t)$ is quasiperiodic in $x$ and in $t$. It is well-known that the $N$-soliton solutions of the KdV equation \eqref{eq:Nsoliton} can be obtained from the Matveev--Its formula by degenerating the hyperelliptic spectral curve to a rational curve with $N$ branch points. Is it possible, conversely, to obtain the Matveev--Its formula \eqref{eq:MatveevIts} as some kind of limit of $N$-soliton solutions \eqref{eq:Nsoliton} when $N\to \infty$?

We may attempt to na\"ively pass to the limit $N\to \infty$ in \eqref{eq:Nsoliton} in the following way. Let $[a,b]$ be an interval on the positive real axis, let $R_1$ be a positive H\"older-continuous function on $[a,b]$, and let $\mu$ be a non-negative measure on $[a,b]$. Consider the following integral equation 
\begin{equation}
f(p,x,t)+\frac{R_1(p)}{\pi}e^{-2px+8p^3t}\int_a^b\frac{f(q,x,t)}{p+q}d\mu(q)=R_1(p)e^{-2px+8p^3t}\label{eq:Ninfinityoneside}
\end{equation}
imposed on a function $f(p,x,t)$, where $p\in [a,b]$. Let $a=\ka_1<\ka_2<\cdots<\ka_N=b$ be a partition of $[a,b]$ uniformly approximating $\mu$. Replacing the above integral with the corresponding Riemann sum, and denoting $c_n=R_1(\ka_n)(b-a)/\pi N$ and $\chi_n=f(\ka_n)(b-a)/\pi N$, we obtain equation \eqref{eq:chiresidues}. Hence equation \eqref{eq:Ninfinityoneside} can be seen as the  limit of \eqref{eq:chiresidues} as $N\to\infty$.

It is easy to show that \eqref{eq:Ninfinityoneside} has a unique solution, and that the corresponding function
\begin{equation}
u(x,t)=-2\frac{d}{dx}\int_a^b f(p,x,t)d\mu(p)
\end{equation}
is a bounded solution of the KdV equation, satisfying the condition $-2b<u<0$. The solution is oscillating as $x\to -\infty$, but as $x\to +\infty$ it is clear that $f(p,x,t)\to R(k)e^{-2kx+8k^3t}$, hence $u(x,t)$ decays exponentially. In other words, $u(x,t)$ can be viewed as a superposition of an infinite number of solitons uniformly bounded away from $+\infty$. In particular, no solution obtained in this way will be an even function of $x$ at any moment of time. It is therefore impossible to obtain the finite-gap solutions given by the Matveev--Its formula \eqref{eq:MatveevIts} in this way, since these solutions are not decreasing as $x\to +\infty$. This lack of symmetry is due to the formulation of the ISM (see Remark~\ref{rem:asymmetry}). These observations were earlier made by Krichever \cite{Krichever}, and a rigorous study of the properties of such solutions, showing the above results, was undertaken by Girotti, Grava and McLaughlin in \cite{GGM18}.

\subsection{Symmetric $N$-soliton solutions} \label{sec:symmetric} In this section, we consider what happens if we try to impose by hand symmetry with respect to the spatial involution $x\mapsto -x$ at $t=0$. We recall than an $N$-soliton solution of the KdV equation \eqref{eq:Nsoliton} is determined by $N$ distinct positive parameters $\ka_1,\ldots,\ka_N$ and $N$ additional positive parameters $q_1,\ldots,q_N$.

\begin{proposition} Let $\ka_1,\ldots,\ka_N$ be distinct positive numbers, and let
\begin{equation}
q_n=\left|\prod_{m\neq n}\frac{\ka_n+\ka_m}{\ka_n-\ka_m}\right|,\quad n=1,\ldots,N. \label{eq:qsymmetric}
\end{equation}
Then the $N$-soliton solution $u(x,t)$ of the KdV equation given by \eqref{eq:Nsoliton} is symmetric at time $t=0$:
\begin{equation}
u(-x,0)=u(x,0).
\end{equation}

\end{proposition}

\begin{proof} At time $t=0$, the function $A(x)=A(x,t)$ is equal to
$$
A(x)=1+q_1e^{-2\ka_1x}+\cdots+q_Ne^{-2\ka_Nx}+\cdots+(q_1\cdots q_N)\prod_{i<j}\frac{(\ka_i-\ka_j)^2}{(\ka_i+\ka_j)^2}e^{-2(\ka_1+\cdots+\ka_N)x}.
$$
Denote $\Phi=\ka_1+\cdots+\ka_N$. We observe that the function $\widetilde{A}(x)=e^{\Phi x}A(x)$ is symmetric: $\widetilde{A}(-x)=\widetilde{A}(x)$. Therefore, so is the corresponding solution of the KdV equation:
$$
u=-2\frac{d^2}{dx^2} \log A=-\frac{d^2}{dx^2} \log \widetilde{A}.
$$

\end{proof}

We now observe that if we attempt to pass to the limit $N\to \infty$, for example by setting $\ka_n=a+(b-a)n/N$, then the coefficients $q_n$ given by \eqref{eq:qsymmetric} have small denominators and diverge. Therefore we cannot obtain finite-gap solutions by this method.

\subsection{From the ISM to the dressing method.} \label{sec:dressing} One of the main results of the papers \cite{DZZ16}, \cite{ZDZ16}, \cite{ZZD16} is a generalization of the ISM within the framework of the dressing method. This construction allows us to take the $N\to \infty$ limit of the set of $N$-soliton solutions and obtain finite-gap solutions. We briefly describe this generalization. 

An $N$-soliton solution is given by Eqs.~\eqref{eq:NsolitonA}-\eqref{eq:Nsoliton}, where the $c_i$ and the $\ka_i$ are the scattering data of a reflectionless potential and are therefore positive. However, formally these equations make sense under the weaker assumption that $\ka_i+\ka_j\neq 0$ for all $i$ and $j$ and that $c_i/\ka_i$ are positive. The corresponding function $\chi$ has poles on both the positive and the negative parts of the imaginary axis.

\begin{proposition} Let $\ka_1,\ldots,\ka_N,c_1,\ldots,c_N$ be nonzero real numbers satisfying the following conditions:  \label{prop:ISMcorrect}
\begin{enumerate}
\item $\ka_i\neq \pm \ka_j$ for $i\neq j$.

\item $c_j/\ka_j>0$ for all $j$. 
\end{enumerate}
Then there exists a unique function $\chi(k,x,t)$ satisfying the following properties:
\begin{enumerate}

\item $\chi$ is meromorphic on the complex $k$-plane with simple poles at the points $i\ka_1,\ldots,i\ka_N$ and no other singularities, and its residues satisfy condition \eqref{eq:residues}. 

\item $\chi$ has the asymptotic behavior \eqref{eq:asymptotic1} as $|k|\to \infty$.

\end{enumerate}

The function $u(x,t)$ given by Eqs.~\eqref{eq:NsolitonA}-\eqref{eq:Nsoliton} is a solution of the KdV equation \eqref{eq:KdV}.
\end{proposition}

We emphasize that, for a given $N$, the set of solutions of the KdV equation obtained using this proposition is still the set of $N$-soliton solutions. Specifically, one can check that the solution given by \eqref{eq:NsolitonA}-\eqref{eq:Nsoliton} for the data $(\ka_1,\ldots,\ka_N,c_1,\ldots,c_N)$ is the $N$-solition solution given by the scattering data $(|\ka_1|,\ldots,|\ka_N|,\widetilde{c}_1,\ldots,\widetilde{c}_N)$, where
\begin{equation}
\widetilde{c}_j=c_j\prod_{\ka_n<0}\left(\frac{\ka_j-\ka_n}{\ka_j+\ka_n}\right)^2\mbox{ if }\ka_j>0,\quad
\widetilde{c}_j=-\frac{4\ka_j^2}{c_j}\prod_{\ka_n<0,\,n\neq j}\left(\frac{\ka_j-\ka_n}{\ka_j+\ka_n}\right)^2\mbox{ if }\ka_j<0.\label{eq:transplant}
\end{equation}
In other words, a $N$-soliton solution with a given set of parameters $\ka_n>0$ and phases $c_n>0$ is described by Prop.~\ref{prop:ISMcorrect} in $2^N$ different ways, by choosing the signs of the $\ka_n$ arbitrarily and adjusting the coefficients $c_n$ using the above formula.

We now give an informal argument why this alternative description of $N$-soliton potentials allows us to obtain finite-gap potentials in the $N\to\infty$ limit. In the previous two sections, we made two attempts to use formulas~\eqref{eq:NsolitonA}-\eqref{eq:Nsoliton} with $\ka_n>0$ to produce $N$-soliton solutions with large $N$. We can either keep the $q_n$ bounded, in which case all solitons end up on the left half-axis, or symmetrically distribute the solitons about $x=0$, in which case the $q_n$ (or, alternatively, the $c_n$) need to be large. 

To obtain a symmetric distribution of $N$ solitons using Proposition~\ref{prop:ISMcorrect}, we choose, as in Section~\ref{sec:symmetric}, a set of parameteres $\ka_n>0$, and set the phases $q_n$ according to~\eqref{eq:qsymmetric}. We then change the signs of half of the $\ka_n$, and change the $c_n$ according to Eq.~\eqref{eq:transplant}. The resulting $c_n$ will be bounded for large $N$, enabling us to take the $N\to \infty$ limit. 


\subsection{Primitive potentials}



In the papers \cite{DZZ16,ZDZ16,ZZD16} the second and third authors considered a contour problem that can be viewed as the limit of Prop.~\ref{prop:ISMcorrect} as $N\to\infty$.

\begin{proposition} Let $0<k_1<k_2$, and let $R_1$ and $R_2$ be positive, H\"older-continuous functions on the interval $[k_1,k_2]$. Suppose that there exists a unique function $\chi(k,x,t)$ satisfying the following properties:

\begin{enumerate} \item $\chi$ is analytic on the complex $k$-plane away from the cuts $[ia,ib]$ and $[-ib,-ia]$ on the imaginary axis, and has non-tangential limits
\begin{equation}
\chi^{\pm}(ip,x,t)=\lim_{\varepsilon\to 0}\chi(ip\pm\varepsilon,x,t),\quad p\in (-k_2,-k_1)\cup(k_1,k_2)\label{eq:limits}
\end{equation}
on the cuts.

\item $\chi$ has jumps on the cuts satisfying
\begin{align}
& \chi^+(ip,x,t)-\chi^-(ip,x,t)=iR_1(p)e^{-2px+8p^3t}[\chi^+(-ip,x,t)+\chi^-(-ip,x,t)], \label{eq:jumps1} \\
&\chi^+(-ip,x,t)-\chi^-(-ip,x,t)=-iR_2(p)e^{2px-8p^3t}[\chi^+(ip,x,t)+\chi^-(ip,x,t)]\label{eq:jumps2}, 
\end{align}
for $p\in[k_1,k_2]$.

\item $\chi$ has asymptotic behavior at infinity
\begin{equation}
\chi(k,x,t)=1+\frac{i}{2k}Q(x,t)+O\left(\frac{1}{k^2}\right),\quad |k|\to \infty,\quad \Im k\neq 0. \label{eq:asymptotic2}
\end{equation}

\item There exist constants $C(x,t)$ and $\alpha<1$ such that near the points $\pm ik_1$ and $\pm ik_2$ the function $\chi$ satisfies
\begin{equation}
|\chi(k,x,t)|< \frac{C(x,t)}{|k\mp ik_j|^{\alpha}},\quad k\to \pm ik_j,\quad j=1,2.
\label{eq:conditionatpoles}
\end{equation}

\end{enumerate}
Then the function $u(x,t)$ given by the formula
\begin{equation}
u(x,t)=\dfrac{d}{dx}Q(x,t) \label{eq:solutionofKdV}
\end{equation}
is a solution of the KdV equation \eqref{eq:KdV}.

\end{proposition}

We call solutions of the KdV equation obtained in this way {\it primitive solutions}. For fixed moments of time, we obtain {\it primitive potentials} of the Schr\"odinger operator \eqref{eq:Sch}.

\begin{remark} Condition \eqref{eq:conditionatpoles} does not appear in the papers \cite{DZZ16,ZDZ16,ZZD16} and is an oversight of the authors. It is necessary, because we consider dressing functions $R_1$ and $R_2$ that do not vanish at $k_1$ and $k_2$. For such functions $\chi$ may have logarithmic or algebraic singularities at the endpoints. Condition \eqref{eq:conditionatpoles} is needed to exclude trivial meromorphic solutions of the Riemann--Hilbert problem, having poles at $\pm ik_j$ and no jump on the cuts. 

We also note that formulas \eqref{eq:jumps1}-\eqref{eq:jumps2} differ from the ones in \cite{DZZ16,ZDZ16,ZZD16} by a factor of $\pi$, this now seems to us to be a more natural normalization of the dressing functions $R_1$ and $R_2$. 

\end{remark}

\begin{remark} 

There is a simple observation that justifies the need to include poles in both the upper and lower half planes when producing a finite gap potential as a limit of $N$-soliton potentials as $N \to \infty$. The spectrum of an $N$-soliton potential determined by $\{\ka_n,c_n\}_{n=1}^N$ is purely simple for the negative energy values $E = -\ka_n^2$, and doubly degenerate for $E>0$. Therefore, a limit as $N\to\infty$ of $N$-soliton solutions with poles in the upper half-plane will have a simple spectrum $E \in [-k_2^2,-k_1^2]$ (in the one band case) and a doubly degenerate spectrum for $E>0$. This is precisely the structure of the spectrum of a one-sided primitive potential having $R_2 \equiv 0$, which limits to a finite gap solution as $x \to - \infty$, but a trivial solution as $x \to \infty$.

A finite-gap potential, on the other hand, has a doubly degenerate continuous spectrum on the interior of its bands, and a simple continuous spectrum on the band ends. To produce a finite-gap potential as a limit of $N$-soliton potentials as $N \to \infty$, we need to include poles in both half-planes, so that in the limit we end up with two linearly independent bounded wave functions for $E$ in the interior of a band.

\end{remark}

A function $\chi(k,x,t)$ satisfying properties \eqref{eq:limits}-\eqref{eq:asymptotic2} can be written in the form
\begin{equation}
\chi(k,x,t)=1+\frac{i}{\pi}\int_{k_1}^{k_2} \frac{f(q,x,t)}{k-iq}dq+\frac{i}{\pi}\int_{k_1}^{k_2}\frac{g(q,x,t)}{k+iq}dq,
\end{equation}
for some functions $f(q,x,t)$ and $g(q,x,t)$ defined for $q\in [a,b]$. Plugging this spectral representation into \eqref{eq:jumps1}-\eqref{eq:jumps2}, we obtain the following system of singular integral equations on $f$ and $g$ for $p\in [k_1,k_2]$:
\begin{align}
f(p,x,t)+\frac{R_1(p)}{\pi}e^{-2px+8p^3t}\left[\int_{k_1}^{k_2} \frac{f(q,x,t)}{p+q}dq+\fint_{k_1}^{k_2}\frac{g(q,x,t)}{p-q}dq\right]=R_1(p)e^{-2px+8p^3t}, \label{eq:f} \\
g(p,x,t)+\frac{R_2(p)}{\pi}e^{2px-8p^3t}\left[\fint_{k_1}^{k_2} \frac{f(q,x,t)}{p-q}dq+\int_{k_1}^{k_2}\frac{g(q,x,t)}{p+q}dq\right]=-R_2(p)e^{2px-8p^3t}. \label{eq:g}
\end{align}
The corresponding solution of the KdV equation is equal to
\begin{equation}
u(x,t)=\frac{2}{\pi}\frac{d}{dx}\int_{k_1}^{k_2}\left[f(q,x,t)+g(q,x,t)\right]dq.
\end{equation}









\section{Symmetric primitive potentials}

In this section, we show how to solve equations \eqref{eq:f}-\eqref{eq:g} analytically as Taylor series in the case when $R_1=R_2$. Suppose that \begin{equation}
R_1(p)=R_2(p)=R(p).
\end{equation}
In this case $g(p,x,t)=-f(p,-x,-t)$ and Eqs.~\eqref{eq:f}-\eqref{eq:g} reduce to the single equation for all $p\in [k_1,k_2]$:
\begin{equation}
f(p,x,t)+\frac{R(p)}{\pi}e^{-2px+8p^3t}\left[\int_{k_1}^{k_2} \frac{f(q,x,t)}{p+q}dq-\fint_{k_1}^{k_2}\frac{f(q,-x,-t)}{p-q}dq\right]=R(p)e^{-2px+8p^3t}.\label{eq:fsym}
\end{equation}
The corresponding primitive solution $u(x,t)$ of the KdV equation
\begin{equation}
u(x,t)=\frac{2}{\pi}\frac{d}{dx}\int_{k_1}^{k_2}\left[f(q,x,t)-f(q,-x,-t)\right]dq
\end{equation}
satisfies the symmetry condition
\begin{equation}
u(-x,-t)=u(x,t).
\end{equation}
In particular, the potential $u(x)=u(x,0)$ at $t=0$ is symmetric:
\begin{equation}
u(-x)=u(x).
\end{equation}
\begin{equation}
u(x)=\frac{2}{\pi}\frac{d}{dx}\int_{k_1}^{k_2}\left[f(q,x)-f(q,-x)\right]dq
\end{equation}
\begin{remark} We emphasize that, in order for a primitive potential to be symmetric, it is sufficient but not necessary for the dressing functions $R_1$ and $R_2$ to be equal.
\end{remark}
We now denote $f(p,x)=f(p,x,0)$ and set $t=0$ in Eq.~\eqref{eq:fsym}:
\begin{equation}
e^{2px}f(p,x)+\frac{R(p)}{\pi}\left[\int_{k_1}^{k_2} \frac{f(q,x)}{p+q}dq-\fint_{k_1}^{k_2}\frac{f(q,-x)}{p-q}dq\right]=R(p),\quad p\in [k_1,k_2].\label{eq:fsymt}
\end{equation}
We show that this equation can be solved analytically. Introduce the variable $s=p^2$ and expand $f(p,x)$ as a Taylor series in $x$, separating the even and odd coefficients in the following way:
\begin{equation}
f(p,x)=\sum_{k=0}^{\infty}\frac{1}{(2k)!}x^{2k}f_k(s)+\sum_{k=0}^{\infty}\frac{1}{(2k+1)!}x^{2k+1}\sqrt{s}h_k(s),\quad s=p^2.
\end{equation}
Plugging this into \eqref{eq:fsymt} and collecting powers of $x$, we obtain the following system of equations on $f_k(s)$ and $h_k(s)$, where $k$ is a non-negative integer:
\begin{equation}
f_k(s)+R(\sqrt{s})H[f_k](s)=R(\sqrt{s})\delta_{0k}-\sum_{i=0}^{k-1}{2k \choose 2i}2^{2k-2i}s^{k-i}f_i(s)-\sum_{j=0}^{k-1}{2k \choose 2j+1}2^{2k-2j-1}s^{k-j}h_j(s), \label{eq:f_k}
\end{equation}
\begin{equation}
h_k(s)-R(\sqrt{s})H[h_k](s)=-\sum_{i=0}^{k}{2k+1 \choose 2i}2^{2k-2i+1}s^{k-i}f_{i}(s)-\sum_{j=0}^{k-1}{2k+1 \choose 2j+1}2^{2k-2j}s^{k-j}h_j(s).\label{eq:h_k}
\end{equation}
Here $H$ is the Hilbert transform on the interval $[k_1^2,k_2^2]$:
\begin{equation}
H[\psi(s)]=\frac{1}{\pi}\fint_{k_1^2}^{k_2^2}\frac{\psi(s')}{s'-s}ds'.
\end{equation}
The corresponding primitive potential is given by
\begin{equation}
u(x)=\frac{2}{\pi}\sum_{k=0}^{\infty} \frac{x^{2k}}{(2k)!}\int_{k_1^2}^{k_2^2} h_k(s')ds'. \label{eq:useries}
\end{equation}
Equations \eqref{eq:f_k}-\eqref{eq:h_k} can be solved recursively for $f_k$ and $h_k$ provided that we know how to invert the operators $1\pm R(\sqrt{s}) H$. This can be done explicitly using the following proposition.

\begin{proposition} \label{prop:inversion} Let $\alpha(s)$ be a H\"older-continuous function on the interval $[k_1^2,k_2^2]$. 
The integral operator $L_{\alpha}$ defined by
\begin{equation}
L_{\alpha}[\psi(s)] = \psi(s) + \tan(\pi \alpha(s)) H[\psi(s)]
\end{equation}
has a unique inverse given by
\begin{equation}
L_{\alpha}^{-1}[\varphi(s)] 
= \cos^2(\pi \alpha(s)) \varphi(s) - \sin(\pi \alpha(s)) e^{-\pi H [\alpha(s)]} H [\cos(\pi \alpha(s)) e^{\pi H [\alpha(s)]} \varphi(s)].
\end{equation}
If $\alpha$ is constant, then $L^{-1}_{\alpha}$ can be written as
\begin{equation}
L_{\alpha}^{-1} [\varphi(s)] 
= \cos^2(\pi \alpha) \varphi(s) - \sin(\pi \alpha) \cos(\pi \alpha) \left( \frac{s - k_1^2}{k_2^2 - s} \right)^{\alpha} H \left[ \left( \frac{k_2^2-s}{s-k_1^2} \right)^{\alpha} \varphi(s) \right].
\end{equation}
\end{proposition}

\begin{proof}
The singular integral equation $L_{\alpha}[\psi(s)] = \varphi(s)$ takes the form
\begin{equation}
\label{eq:intops} \psi(s) - \frac{\tan(\pi \alpha(s))}{\pi} \fint_{k_1^2}^{k_2^2} \frac{\psi(r)}{s-r}dr = \varphi(s).
\end{equation}
We invert this equation to express $\psi$ in terms of $\varphi$ by reformulating it as an inhomogeneous Riemann--Hilbert problem. The function $\Psi(s)$ defined by
\begin{equation*}
\Psi(s) = \frac{1}{\pi} \int_{k_1^2}^{k_2^2} \frac{\psi(r)}{s-r}dr
\end{equation*}
is holomorphic in $s \in \mathbb{C} \setminus [k_1^2,k_2^2]$. The boundary values of $\Psi$ from the right and the left for $s \in [k_1^2,k_2^2]$ satisfy
\begin{equation}
\frac{i}{2} (\Psi^+(s) - \Psi^-(s)) = \psi(s), \quad\frac{1}{2}(\Psi^+ (s) + \Psi^-(s)) = \frac{1}{\pi} \fint_{k_1^2}^{k_2^2} \frac{\psi(r)}{s-r} dr.
\label{eq:Psijumps}
\end{equation}
The integral equation (\ref{eq:intops}) is then equivalent to the Privalov problem
\begin{equation}\label{eq:Priv}
\Psi^+(s) - e^{-2i\pi \alpha(s)} \Psi^-(s) =-2i \cos(\pi \alpha(s)) e^{-i \pi \alpha(s)} \varphi(s)
\end{equation}
where $\Psi$ is normalized by the asymptotic behavior $\Psi(s) \to 0$ as $s \to \infty$. \\

To be able to apply the Plemelj formula to solve the Privalov problem (\ref{eq:Priv}) we first need to remove the multiplicative factor in front of $\Psi^-$. We do this by looking for $\Psi$ in the form $\Psi(s) = \Phi(s) \Xi(s)$. Here the functions $\Phi(s)$ and $\Xi(s)$ are holomorphic in $\mathbb{C} \setminus [k_1^2,k_2^2]$, and satisfy the following conditions.

The function $\Phi(s)$ satisfies the corresponding homogeneous Riemann--Hilbert problem
\begin{equation*} \Phi^+(s) = e^{-2i \pi \alpha(s)} \Phi^-(s) \end{equation*}
and has the asymptotic behavior $\Phi(s) \to 1$ as $s \to \infty$. Such a $\Phi(s)$ is given by
\begin{equation*}
\Phi(s) = \exp\left( \int_{k_1^2}^{k_2^2} \frac{\alpha(r)}{s-r} dr \right).
\end{equation*}
The boundary values of $\Phi$ are
\begin{equation}
\Phi^{\pm}(s) = \exp(-\pi H[\alpha(s)] \mp i \pi \alpha(s)) \label{eq:Phijumps}
\end{equation}
for $s \in [k_1^2,k_2^2]$.
Note that $\Phi \to \Phi^{-1}$ under the transformation $\alpha \to -\alpha$. 

The function $\Xi(s)$ satisfies the jump condition
\begin{equation*}
\Xi^+(s) - \Xi^-(s) = \cos(\pi \alpha(s)) e^{-i\pi \alpha(s)}\frac{-2i \varphi(s)}{\Phi^+(s)} = -2i\cos(\pi \alpha(s))e^{\pi H [\alpha(s)]}\varphi(s)
\end{equation*}
for $s \in [k_1^2,k_2^2]$ and has the asymptotic behavior $\Xi(s) \to 0$ as $s \to \infty$. By the Plemelj formula, $\Xi(s)$ is given by
\begin{equation*}
\Xi(s) = \frac{1}{\pi} \int_{k_1^2}^{k_2^2} \frac{\cos(\pi \alpha(r)) e^{\pi H[\alpha(r)]}\varphi(r)}{s-r } dr =H[\cos(\pi \alpha(s)) e^{\pi H[\alpha(s)]}\varphi(s)].
\end{equation*}
The boundary values of $\Xi$ are
\begin{equation}
\Xi^\pm (s) = H[\cos(\pi \alpha(s)) e^{\pi H[\alpha(s)]}\varphi(s)] \mp i \cos(\pi \alpha(s)) e^{\pi H [\alpha(s)]}\varphi(s)\label{eq:Xijumps}
\end{equation}
for $s \in [k_1^2,k_2^2]$. \\

We now evaluate $\psi(s)$ using \eqref{eq:Psijumps}, \eqref{eq:Phijumps} and \eqref{eq:Xijumps}:
\begin{align*}  \psi(s) = & \frac{i}{2} (\Psi^+(s) - \Psi^-(s))  = \frac{i}{2} (\Phi^+(s) \Xi^+(s) - \Phi^-(s) \Xi^-(s)) \\
 = & \cos^2(\pi \alpha(s)) \varphi(s) - \sin(\pi \alpha(s)) e^{-\pi H [\alpha(s)]} H [\cos(\pi \alpha(s)) e^{\pi H [\alpha(s)]} \varphi(s)],
 \end{align*}
 proving the proposition.
The result for constant $\alpha$ comes from the well-known fact that
\begin{equation} \pi H[1] = \log|s-k_2^2|-\log|s-k_1^2|. \end{equation}
\end{proof}

Using this proposition with $\alpha(s)=\tan^{-1} R(\sqrt{s})/\pi$, we can recursively solve equations \eqref{eq:f_k}-\eqref{eq:h_k} and obtain $u(x)$ as a power series in $x$.

\section{The case of constant $R$}

As an example, we calculate the first two coefficients of $u(x)$ as a Taylor series in the case when $R$ is a constant positive function. Let $\alpha = \tan^{-1}(R)/\pi$, then $0<\alpha<1$. By Prop.~\ref{prop:inversion}, the operators
\begin{equation*}
L_{\pm\alpha} [\psi(s)] = \psi(s) \pm \tan(\pi \alpha) H [\psi(s)]
\end{equation*}
are inverted by
\begin{equation*}
L_{\pm\alpha}^{-1}[\varphi(s)] = \cos^2 (\pi \alpha) \varphi(s) \mp \sin(\pi \alpha) \cos(\pi \alpha) a^{\pm 1}(s) H [a^{\mp 1}(s) \varphi(s)],
\end{equation*}
where the function
\begin{equation}
a(s) = \left(\frac{s-k_1^2}{k_2^2-s}\right)^{\alpha} 
\end{equation}
is continuous on $[k_1^2,k_2^2)$ and has an integrable singularity at $s=k_2^2$.
The equations \eqref{eq:f_k}-\eqref{eq:h_k} determining $f_0,h_0,f_1,h_1$ are 
\begin{align*}
L_{\alpha} [f_0(s)] &= \tan(\pi \alpha), \\
L_{-\alpha} [h_0(s)] &= -2 f_0(s), \\
L_{\alpha} [f_1(s)] &= -4s h_0(s) - 4s f_0(s), \\
L_{-\alpha} [h_1(s)] &= -6 f_1(s) - 12 s h_0(s) - 8 s f_0(s). 
\end{align*}
We compute
\begin{align*}L_{\alpha}^{-1}[1] &= \cos(\pi \alpha) a(s),\\
L_{-\alpha}^{-1}[a(s)] &= \frac{1}{2}(a(s) + a^{-1}(s)),\\
L_{\alpha}^{-1} [sa^{-1}(s)] &= \frac{s}{2} (a(s) + a^{-1}(s)) - \alpha(k_2^2-k_1^2)  a(s),  \\
L_{-\alpha}^{-1} [sa(s)] &= \frac{s}{2} (a(s) + a^{-1}(s)) - \alpha(k_2^2-k_1^2) a^{-1}(s). \end{align*}
We therefore obtain
\begin{align*} f_0(s)  & = \tan(\pi \alpha) L_{\alpha}^{-1}[1] = \sin(\pi \alpha) a(s),  \\
h_0(s) & = -2\sin(\pi \alpha) L_{-\alpha}^{-1}[a(s)] = - \sin(\pi \alpha)(a(s) + a^{-1}(s)), \\
f_1(s) & = 4 \sin(\pi \alpha) L_{\alpha}^{-1}[sa^{-1}(s)]  = 2 \sin(\pi \alpha) s (a(s) + a^{-1}(s)) - 4  \alpha \sin(\pi \alpha) (k_2^2-k_1^2)a(s), \\ 
h_1(s) &= 24 (k_2^2-k_1^2) \alpha \sin(\pi \alpha) L_{-\alpha}^{-1}[a(s)] - 8 \sin(\pi \alpha) L_{-\alpha}^{-1}[sa(s)] \\
 &=(k_2^2-k_1^2)\alpha\sin(\pi\alpha)(12a(s)+20a^{-1}(s))-4\sin(\pi\alpha) s(a(s)+a^{-1}(s)).
\end{align*}
The integrals
\begin{align*} \int_{k_1^2}^{k_2^2} a(s) ds &= \int_{k_1^2}^{k_2^2} a^{-1}(s) ds = \frac{\pi (k_2^2-k_1^2)\alpha}{\sin(\pi \alpha)}, \\
 \int_{k_1^2}^{k_2^2} s a(s) ds &= \frac{\pi\alpha}{2 \sin(\pi \alpha)}( (k_2^4-k_1^4) + \alpha(k_2^2-k_1^2)^2), \\
 \int_{k_1^2}^{k_2^2} s a^{-1}(s) dp &= \frac{\pi \alpha}{2 \sin(\pi \alpha)}((k_2^4-k_1^4) - \alpha(k_2^2-k_1^2)^2), \end{align*}
allow us to compute
\begin{align*} \frac{2}{\pi} \int_{k_1^2}^{k_2^2} h_0(s) ds &= -4  (k_2^2-k_1^2)\alpha, \\
\frac{2}{\pi} \int_{k_1^2}^{k_2^2} h_1(s) ds &=
8(k_2^2-k_1^2)\alpha(4(k_2^2-k_1^2)\alpha - (k_2^2+k_1^2)),
\end{align*}
therefore by Equation~\eqref{eq:useries} we get
\begin{equation} \label{eq:2ndordertaylor} u(x) = -4\alpha(k_2^2-k_1^2) + 4 \alpha(k_2^2-k_1^2)(4 \alpha (k_2^2-k_1^2) - (k_2^2+k_1^2)) x^2 + O(x^4).  \end{equation}
We know that $R=1$ (hence $\alpha=1/4$) and $k_1=0$ produces the exact solution $u(x)=-k_2^2$, and indeed by the above formula we get $u_0 = -k_2^2$ and $u_1 = 0$ in this case. 

Formula \eqref{eq:2ndordertaylor} has some interesting implications.
In the limit as $R \to 0$ we observe that $u(0) \to 0$ and $u''(0) \to 0$.
In the limit as $R \to \infty$ we observe that $u(0) \to -2(k_2^2 - k_1^2)$ and $u''(0) \to 4(k_2^2 - 2 k_1^2)$.
Note that if $k_2^2 > 2 k_1^2$ then $u''(0)$ approaches a positive number from below as $R \to \infty$, but if $k_2^2 < 2 k_1^2$ then $u''(0)$ approaches a negative number.
If $k_2^2 < 2 k_1^2$ we see that in fact $u''(0)$ is negative for all $R$.
On the other hand, if $k_2^2 \ge 2 k_1^2$ then $u''(0)$ will be negative for $R \in (0,\tan(\pi (k_2^2 - k_1^2)/(k_2^1 + k_1^2)))$, $u''(0)$ will be positive for $R \in (\tan(\pi (k_2^2 - k_1^2)/(k_2^1 + k_1^2)),\infty)$, and $u''(0) = 0$ for $R = 0$ or $R = \tan(\pi (k_2^2 - k_1^2)/(k_2^1 + k_1^2))$.


\label{ch:Rconstant}

\section{One-zone symmetric potential}

In this section, we show that the dressing $R_1=R_2=1$ on the interval $[k_1,k_2]$ produces the elliptic one-gap potential
\begin{equation}
u(x)=2\wp(x+i\omega'-\omega)+e_3. \label{eq:uwp}
\end{equation}
Previously, in the papers \cite{ZDZ16,ZZD16}, the second and third authors showed that this potential arises from the dressing
\begin{equation}
R_1(p)=\frac{1}{R_2(p)}=\sqrt{\frac{(q-k_1)(q+k_2)}{(k_2-q)(q+k_1)}}. \label{eq:otherdressing}
\end{equation}
Our new result uses the notation and calculations of \cite{ZDZ16,ZZD16}, but relies on the results of Chapter \ref{ch:Rconstant}.

First, we observe that if
\begin{equation*}
R_2(p)=1/R_1(p),
\end{equation*}
then equations \eqref{eq:jumps1}-\eqref{eq:jumps2} reduce to
\begin{equation*}
\chi^+(ip,x,t)=iR_1(p)e^{-2px+8p^3t}\chi^+(-ip,x,t),\quad \chi^-(ip,x,t)=-iR_1(p)e^{-2px+8p^3t}\chi^-(-ip,x,t),
\end{equation*}
for $p\in[k_1,k_2]$. When $R_1(p)=1$ and $t=0$, the contour problem for $\chi(k,x)=\chi(k,x,0)$ is
\begin{equation}
\chi^+(ip,x)=ie^{-2px}\chi^+(-ip,x),\quad \chi^-(ip,x)=-ie^{-2px}\chi^-(-ip,x),\quad p\in[k_1,k_2].\label{eq:ellipticred}
\end{equation}

Our goal is to find the function $\chi$ satisfying \eqref{eq:ellipticred}. This can in principle be done using the inductive procedure described in Chapter \ref{ch:Rconstant} with $R=1$ and $\alpha=1/4$. However, we will need only the first Taylor coefficient. Indeed, if we set $x=0$, then
\begin{equation*}
f(p,0)=f_0(p)=\sin(\pi \alpha)a(s)=\frac{1}{\sqrt{2}}\left(\frac{s-k_1^2}{k_2^2-s}\right)^{1/4}.
\end{equation*}
Hence we find that the function
\begin{equation*}
\xi(k)=\chi(k,0)=1+\frac{i}{\pi}\int_{k_1}^{k_2}\frac{f(q,0)}{k-iq}dq-\frac{i}{\pi}\int_{k_1}^{k_2}\frac{f(q,0)}{k+iq}dq=
\left(\frac{k^2+k_1^2}{k^2+k_2^2}\right)^{1/4}
\end{equation*}
satisfies equation \eqref{eq:ellipticred} with $x=0$:
\begin{equation}
\xi^+(ip)=i\xi^+(-ip),\quad \xi^-(ip)=-i\xi^-(-ip),\quad p\in[k_1,k_2].\label{eq:chix=0}
\end{equation}

We now look for a solution of \eqref{eq:ellipticred} in the form $\chi(k,x)=\xi(k)\chi_1(k,x)$, where $\chi_1(k,x)$ satisfies the condition
\begin{equation}
\chi_1^+(ip,x)=e^{-2px}\chi_1^+(-ip,x),\quad \chi_2^-(ip,x)=e^{-2px}\chi_2^-(-ip,x),\quad p\in[k_1,k_2].\label{eq:chi1}
\end{equation}
Such a function has already been found in \cite{DZZ16,ZDZ16}. Let $e_1,e_2,e_3$ be defined by the equations
\begin{equation*}
k_1^2=e_2-e_3,\quad k_2^2=e_1-e_3,\quad e_1+e_2+e_3=0. 
\end{equation*}
Let $\wp(z)=\wp(z|\omega,\omega')$ be the Weierstrass function with half-periods $\omega$ and $\omega'$, where $\omega$ is real and $\omega'$ is purely imaginary, such that
\begin{equation*}
e_1=\wp(\omega),\quad e_2=\wp(\omega+i\omega'),\quad e_3=\wp(i\omega').
\end{equation*}
We introduce, as in \cite{DZZ16,ZDZ16}, the variable $z$ via the relation
\begin{equation}
k^2=e_3-\wp(z).
\end{equation}
This relation expresses the complex plane $\mathbb{C}$ with cuts $[ik_1,ik_2]$ and $[-ik_1,-ik_2]$ along the imaginary axis as a double cover of the period rectangle of $\wp$. The Schr\"odinger equation \eqref{eq:Sch} with potential given by \eqref{eq:uwp} is the Lam\'e equation
\begin{equation}
\varphi''-[2\wp(x-\omega-i\omega')+\wp(z)]\varphi=0.
\end{equation}
The Lam\'e equation has a solution
\begin{equation}
\varphi(x,z)=\frac{\sigma(x-\omega-i\omega'+z)\sigma(\omega+i\omega')}{\sigma(x-\omega-i\omega')\sigma(\omega+i\omega'-z)}e^{-\zeta(z)x}
\end{equation}
which has an essential singularity $\varphi(x,z)\sim e^{-x/z}$ near the point $z=0$ (corresponding to $k=\infty$). Therefore the function
\begin{equation}
\chi_1(k,x)=\varphi(x,z)e^{-ikx}=\varphi(x,z)e^{-ix/\sn z}
\end{equation}
tends to 1 as $k\to \infty$. It is easy to check that $\chi_1(k,x)$ satisfies the contour problem \eqref{eq:chi1}. Putting everything together, we obtain the following result.

\begin{proposition} Let $k_2>k_1>0$. Then the function
\begin{equation}
\chi(k,x)=\left(\frac{k^2+k_1^2}{k^2+k_2^2}\right)^{1/4}\varphi(x,z)e^{-ikx},\quad k^2=e_3-\wp(z)
\end{equation} 
satisfies conditions \eqref{eq:limits}-\eqref{eq:conditionatpoles} with $R_1=R_2=1$ and $t=0$. The potential $u(x)$ defined by \eqref{eq:solutionofKdV} is the elliptic one-gap potential \eqref{eq:uwp}.

\end{proposition}

In Section~\ref{sec:dressing}, we observed that an $N$-soliton potential is described using the dressing method in $2^N$ different ways. Since primitive potentials are limits of $N$-soliton potentials, it is also true that a primitive potential can be described using the dressing method in multiple ways, in other words by different pairs of functions $R_1$ and $R_2$. Here we observe an example of this behavior: the elliptic one-gap potential can be constructed using constant dressing functions $R_1=R_2=1$, or using the dressing \eqref{eq:otherdressing}.

\section{Acknowledgments}

The first and third authors gratefully acknowledge the support of NSF grant DMS-1715323. The second author gratefully acknowledges the support of NSF grant DMS-1716822.

\end{document}